\documentclass[amsmath, 12pt]{article}
\usepackage{amsfonts}
\usepackage{amscd,amssymb,amsthm,amsmath}
\usepackage{euscript}
\usepackage{mathrsfs}
\usepackage{euscript}
\usepackage{esvect}
\usepackage{comment}
\usepackage{xcolor}
\usepackage{hyperref}
\usepackage{lmodern}
\usepackage{bm}
\usepackage{graphicx}
\usepackage{booktabs}

\newcommand{\orcid}[1]{\href{https://orcid.org/#1}{\includegraphics[width=8pt]{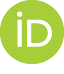}}}

\newtheorem{theorem}{Theorem}

\begin{document}

\title{Characterization of  $n$-Dimensional Toric and Burst-Error-Correcting Quantum Codes from Lattice Codes}

\author{Cibele Cristina Trinca\thanks{The author is with the Department of Biotechnology and Bioprocess Engineering, Federal University of Tocantins, Gurupi-TO, Brazil (e-mail: cibtrinca@yahoo.com.br), \orcid{0000-0003-2857-7410}.}, Reginaldo Palazzo Jr.\thanks{The author is with the School of Electrical and Computer Engineering, State University of Campinas, Brazil (e-mail: palazzo@dt.fee.unicamp.br), \orcid{0000-0002-3219-8061}.}, J. Carmelo Interlando\thanks{The author is with the Department of Mathematics and Statistics, San Diego State University, San Diego, CA, USA (e-mail: carmelo.interlando@sdsu.edu), \orcid{0000-0003-4928-043X}.}, \\Ricardo Augusto Watanabe\thanks{The author is with the School of Mathematics, Statistics and Scientific Computing, State University of Campinas, Brazil (e-mail: ricardoaw18@gmail.com), \orcid{0000-0001-7957-4805}.}, Clarice Dias de Albuquerque\thanks{The author is with the Science and Technology Center, Federal University of Cariri, Juazeiro do Norte, Brazil (e-mail: clarice.albuquerque@ufca.edu.br),\orcid{0000-0003-0673-5334}.}, \\ Edson Donizete de Carvalho\thanks{The author is with the Department of Mathematics, S\~{a}o Paulo State University, Ilha Solteira, SP, Brazil (e-mail: edson.donizete@unesp.br),\orcid{0000-0003-1016-7632}.} and Antonio Aparecido de Andrade\thanks{The author is with the Department of Mathematics, S\~{a}o Paulo State University, Brazil (e-mail: antonio.andrade@unesp.br),\orcid{0000-0001-6452-2236}.}}

\date{\today}
\maketitle


\begin{abstract}
\noindent Quantum error correction is essential for the development of any scalable quantum computer. In this work we introduce a generalization of a quantum interleaving method for combating clusters of errors in toric quantum error-correcting codes. We present new $n$-dimensional toric quantum codes, where $n\geq 5$, which are featured by lattice codes and apply the proposed quantum interleaving method to such new $n$-dimensional toric quantum codes. Through the application of this method to these novel $n$-dimensional toric quantum codes we derive new $n$-dimensional quantum burst-error-correcting codes. Consequently, $n$-dimensional toric quantum codes and burst-error-correcting quantum codes are provided offering both a good code rate and a significant coding gain when it comes to toric quantum codes. Another important consequence from the presented $n$-dimensional toric quantum codes is that if the Golomb and Welch conjecture in \cite{perfcodes} regarding the Lee sphere in $n$ dimensions for the respective close packings holds true, then it follows that these $n$-dimensional toric quantum codes are the only possible ones to be obtained from lattice codes. Moreover, such a methodology can be applied for burst error correction in cases involving localized errors, quantum data storage and quantum channels with memory.
\end{abstract}



\section{Introduction} 

Quantum error correction is essential for the development of any scalable quantum computer. An open question coming from the work \cite{34D} has been the possibility of generalizing the corresponding results to arbitrary dimensions. Accordingly, our contribution in this work is to provide new $n$-dimensional toric quantum codes, where $n\geq 5$, from lattice codes and, after that, apply a quantum interleaving method to such obtained $n$-dimensional toric quantum codes. By applying such a quantum interleaving method to these new codes we present new $n$-dimensional quantum burst-error-correcting codes. 

In \cite{34D} the authors characterize three and four-dimensional toric quantum codes from lattice codes by using the corresponding three and four-dimensional classical single-error-correcting codes from \cite{perfcodes}. From then on, in this work, such a characterization is generalized to the dimensions $n\geq 5$. In \cite{perfcodes} it is conjectured that the classical single-error-correcting code constructed in $n$ dimensions from the $q^{n}$ hypercubic lattice, where $q=2n+1$, which has $q^{n-1}$ codewords featured by the Lee sphere of radius 1 in dimension $n$ is the only case for which a close-packing exists in dimension $n$. Thus, for all $n\geq 5$, the $n$-dimensional toric quantum codes characterized in this work as lattice codes may be the only toric quantum codes from lattice codes in the respective dimension $n$.

Thenceforward, as in \cite{BombinDelgado}, since these new $n$-dimensional toric quantum codes are provided by using the fundamental region of sublattices of the lattice $\mathbb{Z}^{n}$ (lattice codes) that has hypervolume $q$, then the respective code length denoted by $N$ is decreased and, consequently, the corresponding code rate is higher. 

Furthermore, the coding gain of the $n$-dimensional quantum burst-error correcting codes is higher than the coding gain of the $n$-dimensional toric quantum codes and the code rate of the $n$-dimensional quantum burst-error correcting codes is equal to the code rate of the $n$-dimensional toric quantum codes even with the increase in the codeword length of the $n$-dimensional quantum burst-error correcting codes. Also, the larger the dimension $n$, the larger the coding gain of the $n$-dimensional quantum burst-error correcting codes and more burst errors can be corrected.

Moreover, the authors in \cite{34D} present new three and four-dimensional quantum burst-error-correcting codes by using a quantum interleaving method. The quantum interleaving method presented in this work is a generalization to dimension $n$ of the quantum interleaving method from the works \cite{CibeleQIP,34D} and has a different way to interleave the qubits (faces). Consequently, through the quantum interleaving method provided in this work, the coding gain of the three and four-dimensional quantum burst-error-correcting codes from \cite{34D} becomes higher and the corresponding code rate remains the same.    

This paper is organized as follows. Sections II and III review previous results of lattice theory and toric quantum codes, respectively. Section IV furnishes new $n$-dimensional toric quantum codes from lattice codes. In Section V a generalized quantum interleaving method is applied to such new $n$-dimensional toric quantum codes to obtain new $n$-dimensional quantum burst-error-correcting codes.

\section{Background on Lattice Theory}\label{lattice}

A large class of problems in coding theory is related to the properties of lattices \cite{Conway,ijam1,ijam2,ijam3,ClariceArtigo,ClariceQIP,CibeleQIP}. Let $\{ \bm a_1, \ldots, \bm a_n\}$ be a basis for the $n$-dimensional real Euclidean space $\mathbb R^n$, where $n$ is a positive integer. An $n$-dimensional lattice $\Lambda$ is the set of all points of the form $u_1 \bm a_1 + \cdots + u_n \bm a_n$,  where $u_1,\ldots, u_n\in \mathbb{Z}$. Hence, $\Lambda$ is a discrete additive subgroup of $\mathbb R^n$. This property leads to the study of subgroups (sublattices) and coset decompositions (partitions). An algebraic way to obtain sublattices from lattices is via a scaling matrix $A$ with integer entries. Given a lattice $\Lambda$, a sublattice $\Lambda ^{\prime}=A \Lambda$ can be obtained by transforming each vector $\lambda \in \Lambda$ to $\lambda ^{\prime} \in \Lambda ^{\prime}$ according to $\lambda ^{\prime} = A \lambda$.

Every building block that fills the entire space with one lattice point in each region, when repeated many times, is called a \textit{fundamental region} of the lattice $\Lambda$. There are several ways to choose a fundamental region for a lattice $\Lambda$, however the volume of the fundamental region is uniquely determined by $\Lambda$ \cite{ijam1,ClariceArtigo,forney,Conway}. Let $V(\Lambda)$ denote the volume of a fundamental region of the $n$-dimensional lattice $\Lambda$. For a sublattice $\Lambda ^{\prime} = A \Lambda$, we have that $\dfrac{V(A \Lambda)}{V(\Lambda)} = \ \mid \det A \mid$ and the set of the cosets of $\Lambda'$ in $\Lambda$ defines a \textit{lattice code} \cite{forney}.

\section{Background on Toric Quantum Codes}\label{ToricCodes}

The necessary theoretical background related to toric quantum codes can mostly be found in \cite{Kitaev1, Kitaev}. A quantum error-correcting code (QEC) is the image of a linear mapping from the $2^{k}$-dimensional Hilbert space $H^{k}$ to the $2^{n}$-dimensional Hilbert space $H^{n}$, where $k<n$. The codewords are the vectors in the $2^{n}$-dimensional space. The \textit{minimum distance d} of a quantum error-correcting code $C$ is the minimum distance between any two distinct codewords, that is, the minimum Hamming weight of a nonzero codeword. A quantum error-correcting code $C$ of length $n$, dimension $k$ and minimum distance $d$ is denoted by $[[n,k,d]]$. A code with minimum distance $d$ is able to correct up to $t$ errors, where $t=\lfloor \frac{d-1}{2} \rfloor$ \cite{Shor}. 

A stabilizer code $C$ is the simultaneous eigenspace with eigenvalue 1 comprising all the elements of an Abelian subgroup $S$ of the Pauli group $P_{n}$, called the \textit{stabilizer group}. The elements of the Pauli group on $n$ qubits are given by 
\begin{align*}
P_{n}=\{\pm I, \pm iI, \pm X, \pm iX, \pm Y, \pm iY, \pm Z, \pm iZ\}^{\otimes n}, \ \mathrm{where}
\end{align*}
\begin{align*}
I=\left( \begin{array}{cc}
                    1 & 0 \\
                    0 & 1 \\
                    \end{array}
                    \right), \ X=\sigma_{x}=\left( \begin{array}{cc}
                    0 & 1 \\
                    1 & 0 \\
                    \end{array}
                    \right),
\end{align*}
\begin{equation}
Y=\sigma_{y}=\left( \begin{array}{cc}
                    0 & -i \\
                    i & 0 \\
                    \end{array}
                    \right) \ \mathrm{and} \ Z=\sigma_{z}=\left( \begin{array}{cc}
                    1 & 0 \\
                    0 & -1 \\
                    \end{array}
                    \right).
\end{equation}

Thus, $C=\{ \mid\psi \rangle \in H^{n} \ \mid \ M\mid \psi \rangle = \ \mid\psi \rangle, \ \forall \ M\in S \}$ \cite{Gott}.

Kitaev's toric codes form a subclass of stabilizer codes and they are defined in a $q\times q$ square lattice of the torus (Figure 1), where $q$ is a positive integer. Qubits are in one-to-one correspondence with the edges of the square lattice. The parameters of this class of codes are $[[2q^{2},2,q]]$, where the code length $n$ equals the number of edges $|E|=2q^{2}$ of the square lattice. The number of encoded qubits is dependent on the genus of the orientable surface. In particular, codes constructed from orientable surfaces $gT$ (connected sum of $g$ tori $T$) encode $k=2g$ qubits. Thus, codes constructed from the torus, an orientable surface of genus 1, have $k=2$ encoded qubits. The distance is the minimum between the number of edges contained in the smallest homologically nontrivial cycle of the lattice and the number of edges contained in the smallest homologically nontrivial cycle of the dual lattice. Recall that the square lattice is self-dual and a homologically nontrivial cycle is a path of edges in the lattice which cannot be contracted to a face. Therefore the smallest of these two paths corresponds to the orthogonal axes either of the lattice or of the dual lattice. Consequently, $d=q$ \cite{DennisKitaev}.
\begin{figure}[h]%
\centering
\includegraphics[width=0.6\textwidth]{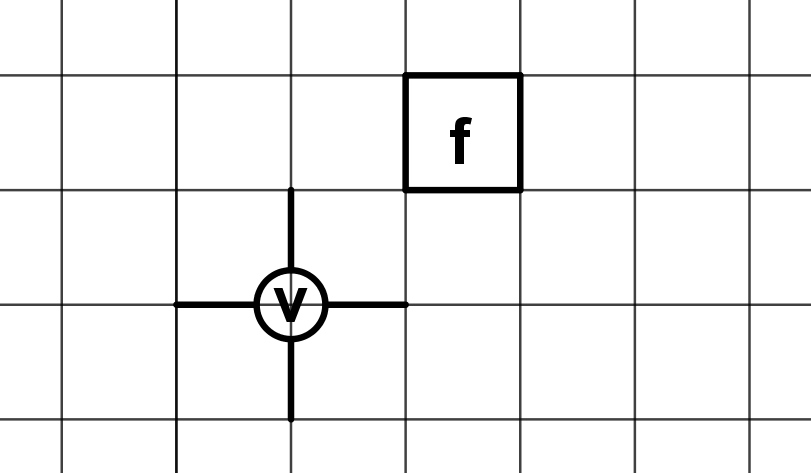}
\caption{Square lattice of the torus, from \cite{ClariceArtigo}.}
\label{fig:1}
\end{figure}

The stabilizer operators are associated with each vertex and each face of the square lattice (lattice) (Figure \ref{fig:1}). Given a vertex $v\in V$, the vertex operator $A_{v}$ is defined by the tensor product of $\sigma_{x}$ -- corresponding to each one of the four edges which have $v$ as a common vertex -- and the operator identity acting on the remaining qubits. Analogously, given a face $f\in F$, the face operator $B_{f}$ is defined by the tensor product $\sigma_{z}$ -- corresponding to each one of the four edges forming the boundary of the face $f$ -- and the operator identity acting on the remaining qubits. In particular, from \cite{ClariceArtigo},  
\begin{equation}    
A_{v}=\bigotimes_{j\in E} \sigma_{x}^{\delta (j\in E_{v})} \ \mathrm{and} \ B_{f}=\bigotimes_{j\in E} \sigma_{z}^{\delta (j\in E_{f})},
\end{equation}
\noindent where $\delta$ is the Kronecker delta. 

The toric code consists of the space fixed by the operators $A_{v}$ and $B_{f}$ and it is given as 
\begin{equation}
C=\{\mid\psi \rangle \in H^{n} \ \mid \ A_{v}\mid\psi \rangle = \ \mid\psi \rangle \ \mathrm{and} \ B_{f} \mid\psi \rangle = \ \mid\psi \rangle , \ \forall \ v,f \}.
\end{equation}

The dimension of $C$ is 4, that is, $C$ encodes $k=2$ qubits.

\section{Characterization of $n$-Dimensional Toric Quantum Codes from  Lattice Codes}\label{ndimensional}

As from \cite{34D} three and four-dimensional toric quantum codes are characterized from lattice codes through the respective three and four-dimensional classical single-error-correcting codes from \cite{perfcodes}. From then on, in this section, such a characterization is generalized to the dimensions $n\geq 5$. In fact, let $n\geq 5$ and $q=2n+1$. The goal of this section is to provide $n$-dimensional toric quantum codes by using the fundamental region of sublattices of the lattice $\mathbb{Z}^{n}$ that has hypervolume $q$, since $q$ divides the hypervolume $q^{n}$ of  $\mathbb{Z}_{q} \times \mathbb{Z}_{q} \times \cdots \times \mathbb{Z}_{q}=\mathbb{Z}_{q}^{n}$ ($n$ times); for the sake of simplicity, we call this region \textit{lattice} or $q^{n}$ \textit{hypercubic lattice}. 

In \cite{perfcodes} the authors construct classical single-error-correcting codes in dimension $n$ in the $q^{n}$ hypercubic lattice which has $q^{n-1}$ codewords. Each codeword of these classical code has the Lee sphere of radius 1 in dimension $n$ as being its fundamental region and the $n$-dimensional hypercube in the center of such Lee sphere corresponds geometrically to the codeword. Therefore each codeword can be featured as an $n$-dimensional central hypercube which has $2 \binom{n}{n-1}$  $(n-1)$-dimensional hypercubes \cite{Coxeter} to which another $n$-dimensional  hypercube has been affixed to each of its $(n-1)$-dimensional hypercubes. In accordance with \cite{perfcodes} it is conjectured that this is the only case for which a close-packing exists in  $n$ dimensions.

Under the algebraic point of view,  the $q^{n}$ hypercubic lattice can be characterized as the group consisting of the cosets of $q \mathbb{Z}^{n}$ in $\mathbb{Z}^{n}$ \cite{Edson}, which in turn is isomorphic to $\mathbb{Z}_{q} \times \mathbb{Z}_{q} \times  \cdots \times \mathbb{Z}_{q} = \mathbb{Z}_{q}^{n}$  ($n$ times), where $q=2n+1$ is a positive integer.  The identifications of the opposite hyperfaces ($(n-1)$-dimensional hypercubes) of the region delimited by $\mathbb{Z}_{q}^{n}$ result in its identification with the $n$-dimensional torus denoted by $T^{n}$ which has genus $g=1$. The hypervolume associated with the lattice $\mathbb{Z}_{q}^{n}$ is $q^{n}$. 

On the condition that each face belongs simultaneously to $2^{(n-2)}$ $n$-dimensional hypercubes of the $q^{n}$ hypercubic lattice (lattice), there are $\binom{n}{2}=\frac{n!}{2!(n-2)!}$ different faces in an $n$-dimensional hypercube (unit cell), that is,  $\binom{n}{2} q^{n}$ different faces in the $q^{n}$ hypercubic lattice (lattice).  In this construction, the qubits are in a biunivocal correspondence with the faces of the $q^{n}$ $n$-dimensional hypercubes of the $q^{n}$ hypercubic lattice. As the $q^{n}$ hypercubic lattice has a finite number of qubits, we use operations modulo $q$ to guarantee that the qubits of a given coset (Lee sphere of radius 1 in $n$ dimensions (fundamental region)) remain inside the $q^{n}$ hypercubic lattice.

The $q^{n}$ hypercubic lattice consists of $q$  $q\times q\times \cdots \times q$ ($(n-1)$ times) = $q^{n-1}$ hypertorus which are $q\times q\times \cdots \times q$ ($(n-1)$ times) = $q^{n-1}$ cross-sections.  It is shown in \cite{ijam,CibeleQIP} that the $q$ codewords of each $q\times q$ cross-section can be generated by the $n$-dimensional  vector $v=(0 \ 0 \ \cdots \ 0 \ 0 \ 1 \ (2n-2))$, that is, we can put them apart by one unit to the right in the horizontal direction and $(2n-2)$ units down in the vertical one.  

Hence, let $j=0,1,2,\cdots,q-1$. Now to rise in the other dimensions to change from one $q\times q$ cross-section to the one above to continue the construction of the $q^{n-2}$ codewords of the corresponding $q^{n-1}$ cross-section we use, respectively, the $n$-dimensional vectors $v_{n-i}$, where $i=n-2,n-3,\ldots,3,2$, as the corresponding generator vectors. Therefore to generate all the $q^{n-2}$ codewords of a $q^{n-1}$ cross-section $j$ (hypertorus) we must use the generator vectors $v_{n-i}$, where $i=n-2,n-3,\ldots,3,2$, and we start the procedure with the null codeword.

Finally to rise in the $n$-th dimension to change from one $q^{n-1}$ cross-section to the other above to continue the construction of the corresponding $q^{n-1}$ codewords we use the $n$-dimensional vector $v_{n-1}=(1 \ 0 \ 0 \ \cdots \ 0 \ 2)$ as the corresponding generator vector. Consequently, by using these $n$-dimensional vectors, we can construct the corresponding $q^{n-1}$ codewords from the $q^{n}$ hypercubic lattice and the $q^{n-1}$ cross-sections can be labeled by the numbers $0,1,2,3,\ldots,q-1$. 

Such generator vectors, that is, $v, v_{2},v_{3},\ldots,v_{n-2},v_{n-1}$, are characterized in the work \cite{ijam}.  By having knowledge of these generator vectors, the following result is obtained.

\begin{theorem}\label{nLatticeCode}
The classical single-error-correcting code from the $q^{n}$ hypercubic lattice, where $q=2n+1$ and $n\geq 5$, can be characterized as a lattice code. 
\end{theorem}  
\begin{proof}
As of \cite{CibeleQIP} the corresponding authors provide for $q=2n+1$, where $n\geq 5$, the $2\times 2$ matrix $\left(
                 \begin{array}{cc}
                   1 & 2n-2 \\
                   1 & -3 \\
                 \end{array}
               \right)$ whose row vectors generate the $q$ codewords of each $q\times q$ cross-section from the $q^{n-1}$ cross-sections (hypertorus). Thus, since the vectors $v_{2},v_{3},\ldots,v_{n-2},v_{n-1}$ are other $n-2$ generator vectors of the classical single-error-correcting code from the $q^{n}$ hypercubic lattice, we complete this matrix with these generators to obtain the $n\times n$ matrix 
\begin{equation}
A=\left(
                 \begin{array}{c}
                   v \\
                   v_{1} \\
                   v_{2} \\
                   v_{3} \\
                   \vdots \\
                   v_{n-2} \\
                   v_{n-1} \\
                 \end{array}
               \right)
\end{equation}           
to the dimension $n$  that, consequently, generates the sublattice $A \mathbb{Z}^{n}$ from $\mathbb{Z}^{n}$, where $v_{1}=(0 \ 0 \ \cdots \ 0 \ 0 \ 1 \ -3)$ is an $n$-dimensional vector. 
               
Now we show that $ \mid \det A \mid =q$. In fact, as already mentioned, the generator vectors $v, v_{2},v_{3},\ldots,v_{n-2},v_{n-1}$ are characterized in the work \cite{ijam}, Section 8, and Table 8 from \cite{ijam} shows such generators for $n=2,3,4,5,6,7,8$.    

Consequently, by permuting the rows of matrix $A$ we obtain matrix 
\begin{equation}
A'=\left(
                 \begin{array}{c}
                   v _{n-1}\\
                   v_{2} \\
                   v_{3} \\
                   \vdots \\
                   v_{n-2} \\
                   v \\
                   v_{1} \\
                 \end{array}
               \right)
\end{equation}               
and, through the characterization of the corresponding generating vectors \cite{ijam}, we have that matrix $A'$ is an upper triangular square matrix in blocks with $A_{1}$ and $A_{2}$ being the respective diagonal blocks, where $A_{1}$ is a $(n-2)\times (n-2)$ square upper triangular matrix whose entries on the diagonal are equal to 1 and $A_{2}=\left(
                 \begin{array}{cc}
                   1 & 2n-2 \\
                   1 & -3 \\
                 \end{array}
               \right)$.  From there, $\mid \det A \mid = \mid \det A' \mid = \mid \det A_{1} . \det A_{2} \mid = 2n+1 = q$. 

Therefore, since $ \mid \det A \mid =q$, $\dfrac{V(A \mathbb{Z}^{n})}{V(\mathbb{Z}^{n})} = \mid \det A \mid = q$.  As $V(\mathbb{Z}^{n})=1$, then $V(A \mathbb{Z}^{n})=q$, that is, the volume of the fundamental region of the sublattice $A \mathbb{Z}^{n}$ of the lattice $\mathbb{Z}^{n}$ is $q$. Thus, it is possible to observe that $\dfrac{V(q \mathbb{Z}^{n})}{V(\mathbb{Z}^{n})}=q^{n}>q=\dfrac{V(A \mathbb{Z}^{n})}{V(\mathbb{Z}^{n})}$ and, consequently, $V(q \mathbb{Z}^{n})=q^{n}>q=V(A \mathbb{Z}^{n})$. 

Hereupon it is possible to show that $A \mathbb{Z}^{n}\supset q \mathbb{Z}^{n}$, in fact: as the codewords of the classical single-error-correcting code in the $q^{n}$ hypercubic lattice are constructed through the sublattice $A\mathbb{Z}^{n}$ modulo $q$, then the null codeword is the equivalence class modulo $q$ which is formed by $q\mathbb{Z}^{n}$. Hence, $A \mathbb{Z}^{n}\supset q \mathbb{Z}^{n}$.  

Now since $V(q \mathbb{Z}^{n})=q^{n}>q=V(A \mathbb{Z}^{n})$ and $A \mathbb{Z}^{n}\supset q \mathbb{Z}^{n}$, then we have the following nested lattice chain
\begin{equation}\label{chain2}
\mathbb{Z}^{n}\supset A \mathbb{Z}^{n}\supset q \mathbb{Z}^{n}.
\end{equation}

From \cite{forney} we obtain $\mid \mathbb{Z}^{n} / q \mathbb{Z}^{n} \mid =\dfrac{V(q \mathbb{Z}^{n})}{V(\mathbb{Z}^{n})}=q^{n}$ and $\mid \mathbb{Z}^{n} / A \mathbb{Z}^{n} \mid=\dfrac{V(A \mathbb{Z}^{n})}{V(\mathbb{Z}^{n})}=q$, consequently, $\mid A \mathbb{Z}^{n} / q \mathbb{Z}^{n} \mid=\dfrac{V(q \mathbb{Z}^{n})}{V(A \mathbb{Z}^{n})}=\dfrac{q^{n}}{q}=q^{n-1}$. Therefore the lattice quotient $A \mathbb{Z}^{n} / q \mathbb{Z}^{n}$ is the group consisting of the $q^{n-1}$ cosets of $q \mathbb{Z}^{n}$ in $A \mathbb{Z}^{n}$ and the set of these $q^{n-1}$ cosets defines a lattice code. 

Now as the row vectors of the matrix $A$ which generates the sublattice $A \mathbb{Z}^{n}$ are the generators of the classical code from the $q^{n}$ hypercubic lattice and the lattice quotient $A \mathbb{Z}^{n} / q \mathbb{Z}^{n}$ provides operations modulo $q$ over the respective vectors, then there exists a natural group isomorphism between the classical code and the lattice code $A \mathbb{Z}^{n} / q \mathbb{Z}^{n}$, consequently, such a classical code is characterized by the lattice code $A \mathbb{Z}^{n} / q \mathbb{Z}^{n}$. 
\end{proof}

Through these lattice codes, where $n\geq 5$,  new $n$-dimensional toric quantum codes can be characterized; in fact, since the $\mathbb{Z}^{n}$-lattice is self-dual, the qubits are in a biunivocal correspondence with the faces of the $q^{n}$ $n$-dimensional hypercubes of the $q^{n}$ hypercubic lattice and one associates a $X$-parity check acting on cubes, and a $Z$-parity check acting on edges. Thenceforward, as in \cite{BombinDelgado}, the respective code length denoted by $N$ is decreased and it is given by the number of faces of the Lee sphere of radius 1 in $n$ dimensions, that is, $N= \binom{n}{2} q$, since such a Lee sphere has hypervolume $q=2n+1$ and there are $\binom{n}{2}=\frac{n!}{2!(n-2)!}$ different faces in an $n$-dimensional hypercube which is a unit cell. 

The code dimension is $k=\binom{n}{2}$, since these $n$-dimensional toric quantum codes are constructed on the $T^{n}$ torus which has genus $g=1$ \cite{Kitaev}. As the $\mathbb{Z}^{n}$-lattice is self-dual, then the code distance is defined as the minimum number of faces in the $q^{n}$ hypercubic lattice between two codewords and, therefore, such a code distance is characterized as the minimum \textit{Mannheim distance} which is given by $d_{M} = min \{ \mid x_{1} \mid + \mid x_{2} \mid + \cdots + \mid x_{n} \mid \ \ \mid \ (x_{1},x_{2},\ldots,x_{n}) \in \mathcal{C} \}$, where $\mathcal{C}$ indicates the set of the $q^{n-1}$ codewords and $w_{M} (x_{1},x_{2},\ldots,x_{n}) = \mid x_{1} \mid + \mid x_{2} \mid + \cdots + \mid x_{n} \mid$ is known as the \textit{Mannheim weight} of $(x_{1},x_{2},\ldots,x_{n})$. 

\begin{theorem}\label{ndistance}
The minimum Mannheim distance of the respective $n$-dimensional toric quantum codes is given by $3$.
\end{theorem}
\begin{proof}
From \cite{CibeleQIP}, the row vectors of the $2\times 2$ matrix $\left(
                 \begin{array}{cc}
                   1 & 2n-2 \\
                   1 & -3 \\
                 \end{array}
               \right)$ generate the $q$ codewords of each $q\times q$ cross-section from the $q^{n-1}$ cross-sections (hypertorus). On the other hand, in \cite{CibeleQIP}, Theorem 2 provides the minimum Mannheim distance of the codewords of each $q\times q$ cross-section which is given by 4, for $n\geq 5$.   

As we have seen previously we use the vectors $v_{2},v_{3},\ldots,v_{n-2},v_{n-1}$ which are the corresponding generators in the vertical directions to rise in the three, four, five up to $n$ dimensions, respectively, that is, to change from one cross-section to the other above, respectively, to continue the construction of the corresponding $q^{n-1}$ codewords.         

As already mentioned, the generator vectors $v_{2},v_{3},\ldots,v_{n-2},v_{n-1}$ are characterized in the work \cite{ijam}, Section 8, and Table 8 from \cite{ijam} shows such generators for $n=2,3,4,5,6,7,8$. From there, the last coordinate of the vector $v_{n-2}$ is given by $2n-7\geq 3$, for $n\geq 5$. Furthermore,  for vectors $v_{n-i}$, where $i=n-2,n-3,\ldots,3,2$, their coordinates at positions between $n-i$ and $n-1$ are equal to 1 and, consequently, their minimum Mannheim weight is greater than or equal to 3. 

From then on, as long as the minimum Mannheim distance of the codewords of each $q\times q$ cross-section is given by 4 and the minimum Mannheim weight of $v_{2},v_{3},\ldots,v_{n-2},v_{n-1}$ is 3, then the minimum Mannheim distance of the corresponding $n$-dimensional toric quantum codes is given by $3$.
\end{proof}

As a consequence, the parameters of the new $n$-dimensional toric quantum codes, where $n\geq 5$, are $[[N= \binom{n}{2} q,k=\binom{n}{2},d_{M}=3]]$.

Consequently, under the algebraic point of view, such new $n$-dimensional toric quantum codes can be characterized as the group consisting of the cosets of $q \mathbb{Z}^{n}$ in $A \mathbb{Z}^{n}$. 

As the $\mathbb{Z}^{n}$-lattice is self-dual, then the respective quantum channel without memory is symmetric since the duality of the lattice in which a toric quantum code is constructed influences the error correction pattern \cite{LivroBombin}.

\section{New $n$-Dimensional Quantum Burst-Error-Correcting Codes}

The works \cite{CibeleQIP,34D} present new toric quantum codes and a quantum interleaving method which is used over these new codes in 2 and, 3 and 4 dimensions, respectively. Moreover, the quantum interleaving method coming from \cite{34D} is a generalization to dimensions 3 and four of the quantum interleaving method from the work \cite{CibeleQIP}.  

In this way, by continuing these two works \cite{CibeleQIP,34D}, the corresponding quantum interleaving methods proposed in \cite{CibeleQIP,34D}, respectively, are generalized in this work to dimensions $n\geq 5$ and the corresponding quantum interleaving method presented in this section to dimension $n$ has a different way to interleave the qubits (faces). Thus, by applying it to the new $n$-dimensional toric quantum codes from Section \ref{ndimensional}, we obtain new $n$-dimensional quantum burst-error correcting codes which have higher code rate and coding gain when it comes to toric quantum codes. Also, this section presents and discusses the relevant comparisons which are organized in Tables \ref{T1} and \ref{T2}. Next Theorem provides the parameters of these new $n$-dimensional quantum burst-error correcting codes.

\begin{theorem}
Consider $n\geq 5$ and $q=2n+1$, where $n$ is the dimension of the $q^{n}$ hypercubic lattice. The combination of the $n$-dimensional toric quantum codes obtained in Section \ref{ndimensional} and the interleaving technique results in $n$-dimensional interleaved toric quantum codes with parameters $[[\alpha q^{n},\alpha q^{n-1},t_{i}=q^{2}]]$, where $\alpha=\binom{n}{2}$ and $t_{i}$ is the interleaved toric quantum code error correcting capability.
\end{theorem}
\begin{proof}
The parameters of the $n$-dimensional toric quantum codes, where $n\geq 5$, presented in Section \ref{ndimensional} are given by $[[N= \alpha q,k=\alpha,d_{M}=3]]$, where $\alpha=\binom{n}{2}$. A toric quantum code with minimum distance $d_{M}$ is able to correct up to $t$ errors, where $t=\lfloor \frac{d_{M}-1}{2} \rfloor$ \cite{ClariceArtigo}, therefore, such codes are able to correct up to $t=1$ error. 

Note that the clusters of errors have the shape of the Lee sphere of radius 1 in $n$ dimensions and the qubits are in a biunivocal correspondence with the faces of the $q^{n}$ hypercubic lattice. Figure \ref{fig9} shows the representation of the storage system under analysis.
\begin{figure}[h]%
\centering
\includegraphics[width=0.9\textwidth]{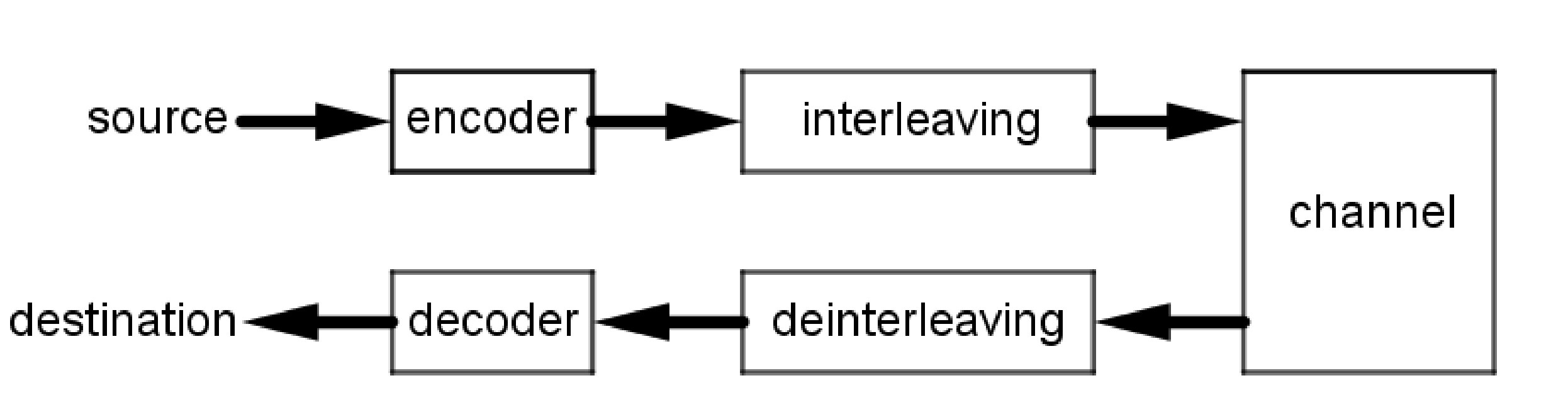}
\caption{Representation of an storage system \cite{34D}.}
\label{fig9}
\end{figure}

Consider the $q^{n}$ hypercubic lattice, where $n\geq 5$ and $q=2n+1$. Such a hypercubic  lattice has a total of $\alpha q^{n}$ qubits (faces), where $\alpha=\binom{n}{2}$. From now on we explain how to spread such $\alpha q^{n}$ adjacent qubits in the $q^{n}$ hypercubic lattice. Observe that the $n$-dimensional toric quantum codes characterized in Section \ref{ndimensional} from lattice codes in the $q^{n}$ hypercubic lattice consist of $q^{(n-1)}$ codewords whose fundamental region is the Lee sphere of radius 1 in $n$ dimensions which has $q$ hypercubes and $\alpha q$ qubits.   

Thereby the first $\alpha q^{(n-1)}$ adjacent qubits which are related to the codewords of the cross-section $j=0$ are spread over the $q^{(n-2)}$ codewords of the cross-section $j=0$ as it follows: firstly, ordering such adjacency consists of using the order of the codewords of the cross-section $j=0$ where such an order is featured by the generator vectors that are used in the generation of the codewords of each cross-section $j$ ($j=0,1,\ldots,q-1$); observe that these generator vectors are the same for generating the corresponding codewords of each cross-section $j$. Thenceforward, the $\alpha q$ qubits of the null codeword are placed on the hypercubes that correspond to the first $q$ codewords of the $q^{(n-2)}$ codewords of the cross-section $j=0$, the first $q$ qubits of the null codeword are placed on the hypercubes that correspond to the first $q$ codewords of the $q^{(n-2)}$ codewords of the cross-section $j=0$ and the other $\alpha-1$ blocks of $q$ qubits are placed in the same way on the same hypercubes (related to the first $q$ codewords of the $q^{(n-2)}$ codewords of the cross-section $j=0$).   

Now for each $\alpha q$ qubits of the next $q^{(n-3)}-1$ codewords of the cross-section $j=0$, these $\alpha q$ qubits are placed in an analogous manner on the hypercubes that correspond to the next $q$ codewords of the $q^{(n-2)}$ codewords of the cross-section $j=0$. Thus, it is observed that the $\alpha q^{(n-2)}$ qubits of the first $q^{(n-3)}$ codewords of the cross-section $j=0$ are spread exactly on the hypercubes that correspond to the $q^{(n-2)}$ codewords of the cross-section $j=0$. 

Therewith, as the corresponding fundamental region is the Lee sphere of radius 1 in $n$ dimensions which has $q$ hypercubes by including the one related to the codeword, performing such interleaving procedure for the next $q-1$ blocks of $q^{(n-3)}$ codewords of the cross-section $j=0$ and starting the corresponding arrangement by placing, for each block $q-1$, the first qubit of the $\alpha q^{(n-2)}$ qubits on one of the $q-1$ hypercubes related to the Lee sphere of radius 1 of the null codeword (the hypercube related to the null codeword was used to spread the qubits of the first block) and, then, we use the generator lattice vectors of the codewords of each cross-section $j$ to follow the sequence of the respective codewords of the cross-section $j=0$ to spread the $\alpha q^{(n-2)}$ qubits exactly on the same hypercube of these codewords. 

Now the other $q-1$ blocks with $\alpha q^{(n-1)}$ adjacent qubits which are related to the codewords of each $q^{n-1}$ cross-section $j=1,2,\ldots,q-1$, respectively, are spread analogously across the $q^{(n-2)}$ codewords of each cross-section $j$. Also, observe that the codeword of  each $q^{n-1}$ cross-section $j=1,2,\ldots,q-1$ which starts the spreading of the respective $\alpha q^{(n-1)}$ adjacent qubits is that featured by the $n$-dimensional generator lattice vector $v_{n-1}=(1 \ 0 \ 0 \ \cdots \ 0 \ 2)$ which is used to rise in the $n$-th dimension to change from one $q^{n-1}$ cross-section to the other above in the construction of  the $q^{n-1}$ codewords from the $q^{n}$ hypercubic lattice. 

Therefore, in this way, the $\alpha q^{n}$ adjacent qubits are spread across the $q^{n}$ hypercubic lattice. Notice that it is needed to follow the sequence of the codewords of the respective code constructed over the $q^{n}$ hypercubic lattice by using all the generator lattice vectors. 

As the $\alpha q^{(n-2)}$ qubits of every block of $q^{(n-3)}$ codewords of each $q^{n-1}$ cross-section $j$ are spread exactly over the same hypercube of the Lee spheres of radius 1 in $n$ dimensions that correspond to the $q^{(n-2)}$ codewords of the $q^{n-1}$ cross-section $j$, then, if we assume that a cluster of $q$ errors in each $q^{n-1}$ cross-section $j$ has the shape of the Lee sphere of radius 1 in $n$ dimensions, since each $q^{n-1}$ cross-section has $q$ blocks of $q^{(n-3)}$ codewords and the $n$-dimensional toric quantum codes presented in Section \ref{ndimensional} are able to correct up to $t=1$ error, such an interleaving technique shows that when the deinterleaving process is applied each one of these $q$ errors occurs in a different codeword of the $q^{n-1}$ cross-section $j$ and then the respective $n$-dimensional toric quantum code is applied to correct these $q$ errors in burst. Consequently, as $j=0,1,2,\cdots,q-1$, such an interleaving technique shows that up to $q^{2}$ errors in burst can be corrected over the $q^{n}$ hypercubic lattice. 

Accordingly, in this section, it is shown that the combination of the $n$-dimensional toric quantum codes provided in Section \ref{ndimensional} whose parameters are $[[N= \alpha q,k=\alpha,d_{M}=3]]$, where $\alpha=\binom{n}{2}$,  and the presented interleaving method results in $n$-dimensional interleaved toric quantum codes whose parameters are given by $[[\alpha q^{n},\alpha q^{n-1},t_{i}=q^{2}]]$, where $t_{i}$ is the interleaved toric quantum code error correcting capability. As a consequence, new $n$-dimensional quantum burst-error correcting codes are obtained by applying this interleaving method to the respective $n$-dimensional toric quantum codes. \end{proof}

The code rate \cite{rates} and the coding gain \cite{rates} are given by $R=\dfrac{k}{N}$ and $G=\dfrac{k}{N}(t+1)$, respectively, where $t$ is the toric quantum code error correcting capability. 

In \cite{perfcodes} it is conjectured that the classical single-error-correcting code constructed in $n$ dimensions from the $q^{n}$ hypercubic lattice which has $q^{n-1}$ codewords featured by the Lee sphere of radius 1 in dimension $n$ is the only case for which a close-packing exists in dimension $n$. Thus, for all $n\geq 5$, the $n$-dimensional toric quantum codes characterized as lattice codes in Section \ref{ndimensional} may be the only toric quantum codes from lattice codes in the respective dimension $n$.    

Table \ref{T1} shows the code rate and the coding gain (in dB) of the $n$-dimensional toric quantum codes provided in Section \ref{ndimensional}.  Thenceforward, as in \cite{BombinDelgado}, since the $n$-dimensional toric quantum codes are provided by using the fundamental region of sublattices of the lattice $\mathbb{Z}^{n}$ (lattice codes) that has hypervolume $q$, then the respective code length denoted by $N$ is decreased and, consequently, the corresponding code rate is higher.

\begin{table}[h!]
\begin{center}
\begin{minipage}{174pt}
\caption{Code rate and coding gain of the $n$-dimensional toric quantum codes from Section \ref{ndimensional}} \label{T1}%
\begin{center}
\begin{tabular}{@{}lll@{}}
\toprule
$n$-D Toric Quantum Codes& Code Rate & Coding Gain \\
$[[N= \alpha q,k=\alpha,d_{M}=3]]$ & $R=\dfrac{1}{q}=\dfrac{1}{2n+1}$ & $G=\dfrac{2}{q}=\dfrac{2}{2n+1}$, \mbox{dB} \\
\midrule
$[[110,k=10,3]]$ ($n=5$) & 0.09091 & 0.18182 \\
$[[195,k=15,3]]$ ($n=6$) & 0.07692 & 0.15385 \\ 
$[[315,k=21,3]]$ ($n=7$) & 0.06667 & 0.13333 \\
$[[476,k=28,3]]$ ($n=8$) & 0.05882 & 0.11765 \\
\bottomrule
\end{tabular}
\end{center}
\end{minipage}
\end{center}
\end{table}

Table \ref{T2} shows the equivalent interleaved $n$-dimensional toric quantum codes ($n$-dimensional quantum burst-error-correcting codes) and their corresponding interleaving code rates $R_{i}$ and coding gains $G_{i}$ (in dB) when a cluster of errors in each $q^{n-1}$ cross-section $j$ has the shape of the corresponding Lee sphere of radius 1 over the $q^{n}$ hypercubic lattice.  

Notice that the coding gain of the $n$-dimensional quantum burst-error correcting codes is higher than the coding gain of the $n$-dimensional toric quantum codes provided in Section \ref{ndimensional} and the code rate of the $n$-dimensional quantum burst-error correcting codes is equal to the code rate of the $n$-dimensional toric quantum codes provided in Section \ref{ndimensional} even with the increase in the codeword length of the $n$-dimensional quantum burst-error correcting codes. Note that the larger $n$, the larger $q$, and therefore the larger the dimension $n$, the larger the coding gain $G_{i}=q+\dfrac{1}{q}$ of the $n$-dimensional quantum burst-error correcting codes and more burst errors can be corrected.

Moreover, the authors in \cite{34D} present new three and four-dimensional quantum burst-error-correcting codes by using a quantum interleaving method. As mentioned previously, the quantum interleaving method presented in this section is an extension to dimension $n$ of the quantum interleaving method from the works \cite{CibeleQIP,34D} and has a different way to interleave the qubits (faces). Consequently, through the quantum interleaving method provided in this section, the coding gain of the three and four-dimensional quantum burst-error-correcting codes from \cite{34D} becomes higher and the corresponding code rate remains the same.

\begin{table}[h!]
\begin{center}
\begin{minipage}{174pt}
\caption{Code rate and coding gain of the $n$-dimensional quantum burst-error-correcting codes from the interleaving method}\label{T2}%
\begin{center}
\begin{tabular}{@{}lll@{}}
\toprule
Interleaved $n$-D & Code Rate & Coding Gain \\
Toric Quantum Codes & $R_{i}=\dfrac{k}{N}=\dfrac{1}{q}$ & $G_{i}=\dfrac{k}{N}(t_{i}+1)$, \mbox{dB} \\
\midrule
$[[10 q^{5},10 q^{4},t_{i}=q^{2}]]$ ($n=5$, $q=11$) & 0.09091 & 11.09102 \\
$[[15 q^{6},15 q^{5},t_{i}=q^{2}]]$ ($n=6$, $q=13$) & 0.07692 & 13.0764 \\ 
$[[21 q^{7},21 q^{6},t_{i}=q^{2}]]$ ($n=7$, $q=15$) & 0.06667 & 15.06742 \\
$[[28 q^{8},28 q^{7},t_{i}=q^{2}]]$ ($n=8$, $q=17$) & 0.05882 & 17.0578 \\ 
\bottomrule
\end{tabular}
\end{center}
\end{minipage}
\end{center}
\end{table}

\section{Conclusion}

In this work it is shown that the combination of the $n$-dimensional toric quantum codes characterized by lattice codes and the provided quantum interleaving method results in $n$-dimensional interleaved toric quantum codes which are new $n$-dimensional quantum burst-error correcting codes and, to the best of our knowledge, little is known regarding interleaving techniques for combating cluster of errors in toric quantum codes. Besides, a single quantum burst-error-correcting code in $n$ dimensions that has a good code rate and coding gain when it comes to toric quantum codes is needed to correct the respective errors in burst, while in the classical interleaving techniques \cite{ijam,celso} several classical error-correcting codes are needed to correct the corresponding burst errors. The quantum interleaving method from this work is a generalization of those from the works \cite{CibeleQIP,34D} to dimension $n$ and it has a different way to interleave the qubits (faces) which furnishes a higher coding gain and error correcting capability. The quantum interleaving in both this work and in \cite{CibeleQIP,34D}, the qubits are allocated in each fundamental region of the $n$-dimensional toric quantum code and not sequentially, as occurs in classical interleaving \cite{ijam,celso}.

\section{Acknowledgments}

The authors would like to thank the financial Brazilian agency CNPq (Conselho Nacional de Desenvolvimento Cient\'{i}fico e Tecnol\'{o}gico - Brazil) for the funding support and under grants no. 101862/2022-9 and no. 305239/2020-1.

\end{document}